\documentclass[12pt]{article}
\usepackage{amsthm}
\usepackage{amsmath}
\usepackage{amssymb}
\usepackage{amsfonts}
\usepackage{graphicx}
\usepackage{latexsym}
\usepackage{amssymb}
\usepackage{setspace}
\usepackage{comment}
\usepackage{color}
\usepackage[noblocks]{authblk}

\newcommand{\cB}{{\cal B}}

\newcommand{\cG}{{\cal G}}

\newcommand{\cS}{{\cal S}}

\newcommand{\cF}{{\cal F}}
\newcommand{\cC}{{\cal C}}
\newcommand{\cP}{{\cal P}}

\renewcommand{\leq}{\leqslant}

\renewcommand{\geq}{\geqslant}

\newcommand{\Cref}[1]{Co\-rol\-la\-ry\,\ref{#1}}

\newcommand{\sP}{\cP}
\newcommand{\sG}{\cG}

\newcommand{\Gr}{\smash{{\sG\kern-1.5pt}_q\kern-0.5pt(n,k)}}
\newcommand{\Grr}{\smash{{\sG\kern-1.5pt}_q\kern-0.5pt(n,k+1)}}
\newcommand{\Gk}{\smash{{\sG\kern-1.5pt}_q\kern-0.5pt(n,k_1)}}
\newcommand{\Gkk}{\smash{{\sG\kern-1.5pt}_q\kern-0.5pt(n,k_2)}}
\newcommand{\Grtwo}{\smash{{\sG\kern-1.5pt}_2\kern-0.5pt(n,k)}}
\newcommand{\Gkone}{\smash{{\sG\kern-1.5pt}_q\kern-0.5pt(n,k_1)}}
\newcommand{\Gktwo}{\smash{{\sG\kern-1.5pt}_q\kern-0.5pt(n,k_2)}}
\newcommand{\Ps}{\smash{{\sP\kern-2.0pt}_q\kern-0.5pt(n)}}

\newcommand{\linadd}{\kern1pt\mbox{\small$\boxplus$}\kern1pt}

\newtheorem{definition}{Definition}
\newtheorem{theorem}{Theorem}
\newtheorem{lemma}{Lemma}
\newtheorem{cor}{Corollary}

\newtheorem{example}{Example}
\newtheorem{construction}{Construction}
\newtheorem{remark}{Remark}
\newtheorem{claim}{Claim}

\begin{document}

\title{Multiset Combinatorial Batch Codes}

\author[1]{Hui Zhang\thanks{The work of the first author is supported in part at the Technion by a fellowship of the Israel Council of Higher Education.}}
\author[1]{Eitan Yaakobi}
\author[2]{Natalia Silberstein}

\affil[1]{\small Computer Science Department, Technion---Israel Institute of Technology, Haifa, Israel. \par email: huizhang@cs.technion.ac.il, yaakobi@cs.technion.ac.il}
\affil[2]{\small Yahoo! Labs, Haifa, Israel. email: natalys@cs.technion.ac.il}


\date{}
\maketitle
\begin{abstract}
\emph{Batch codes}, first introduced by Ishai, Kushilevitz, Ostrovsky, and Sahai, mimic a distributed storage of a set of $n$ data items on $m$ servers, in such a way that any batch of $k$ data items can be retrieved by reading at most some $t$ symbols from each server. \emph{Combinatorial batch codes}, 
are replication-based batch codes in which each server stores a subset of the data items.

In this paper, we propose a generalization of combinatorial batch codes, called \emph{multiset combinatorial batch codes} (\emph{MCBC}), in which $n$ data items are stored in $m$ servers, such that any multiset request of $k$ items, where any item is requested at most  $r$ times, can be retrieved by reading at most $t$ items from each server. The setup of this new family of codes is motivated by recent work on codes which enable high availability and parallel reads in distributed storage systems. 
The main problem under this paradigm is to minimize the number of items stored in the servers, given the values of $n,m,k,r,t$, which is denoted by $N(n,k,m,t;r)$. 
We first give a necessary and sufficient condition for the existence of MCBCs. Then, we present several bounds on $N(n,k,m,t;r)$ and constructions of MCBCs. In particular, we determine the value of $N(n,k,m,1;r)$ for any $n\geq \left\lfloor\frac{k-1}{r}\right\rfloor{m\choose k-1}-(m-k+1)A(m,4,k-2)$, where $A(m,4,k-2)$ is the maximum size of a binary constant weight code of length $m$, distance four and weight $k-2$. We also determine the exact value of $N(n,k,m,1;r)$ when $r\in\{k,k-1\}$ or $k=m$.
\end{abstract}

\section{Introduction}\label{sec:BatchCodes}

\subsection{Background and Definitions}

Batch codes were first introduced by Ishai \textit{et al.} in~\cite{IKOS04} as a method to represent the distributed storage of a set of $n$ data items on $m$ servers. These codes were originally motivated by several applications such as load balancing in distributed storage, private information retrieval, and cryptographic protocols. Formally, these codes are defined as follows~\cite{IKOS04}.
\begin{definition}
\begin{enumerate}
\item  An $(n,N,k,m,t)$ {\em batch code} over an alphabet $\Sigma$, encodes a string $x\in \Sigma^n$ into an $m$-tuple of strings $y_1,\ldots,y_m\in \Sigma^*$ (called {\em buckets} or {\em servers}) of total length $N$, such that for each $k$-tuple (called {\em batch} or {\em request}) of distinct indices $i_1,\ldots,i_k\in[n]$, the $k$ data items $x_{i_1},\ldots,x_{i_k}$ can be decoded by reading at most $t$ symbols from each server.
\item An $(n,N,k,m,t)$ {\em multiset batch code} is an $(n,N,k,m,t)$ batch code which also satisfies the following property: For any multiset request of $k$ indices $i_1,\ldots,i_k\in[n]$ there is a partition of the buckets into $k$ subsets $S_1,\ldots,S_k \subseteq [m]$ such that each item $x_{i_j}, j\in[k]$, can be retrieved by reading at most $t$ symbols from each bucket in $S_j$. 
\end{enumerate}
\end{definition}



Yet another class of codes, called \emph{combinatorial batch codes} (CBC), is a special type of batch codes in which all encoded symbols are copies of the input items, i.e., these codes are replication-based. Several works have considered codes under this setup; see e.g.~\cite{BB14,BRR12,BKMS10,BT11C,BT11A,BT11B,BT12,BT15,PSW09,S14,SG13}. However, note that combinatorial batch codes are not multiset batch codes and don't allow to request an item more than once. 

Motivated by the works on codes which enable parallel reads for different users in distributed storage systems, for example,  codes with locality and availability \cite{RPDV2016,ZY2016}, we introduce a generalization of CBCs, named \emph{multiset combinatorial batch codes}. 

\begin{definition}
An $(n,N,k,m,t;r)$ {\em multiset combinatorial batch code (MCBC)} is a collection of subsets of $[n]$, $\cC=\{C_1,C_2,\dots,C_m\}$ (called {\em servers}) where $N=\sum_{j=1}^m |C_j|$, such that for each multiset request $\{i_1,i_2,\dots,i_k\}$, in which every element in $[n]$ has multiplicity at most $r$, there exist subsets $D_1,\ldots,D_m$, where for all $j\in[m]$, $D_j\subseteq C_j$ with $|D_j|\leq t$, and the multiset union\footnote{For any $i\in[n]$, the multiplicity of $i$ in the multiset union of the sets $D_j$ for $j\in[m]$ is the number of subsets that contain $i$, that is $|\{j\in[m]:i\in D_j\}|$.} of $D_j$ for $j\in[m]$ contains the multiset request $\{i_1,i_2,\dots,i_k\}$.
\end{definition}

In other words, an $(n,N,k,m,t;r)$-MCBC is a coding scheme which encodes $n$ items into $m$ servers, with total storage of $N$ items, such that any multiset request of items of size at most $k$, where any item can be repeated at most $r$ times, can be retrieved by reading at most $t$ items from each server. In particular, when $r=1$ we obtain a combinatorial batch code, and when $r=k$ and $t=1$ we obtain a multiset batch code based on replication.


\begin{example}
Let us consider the following $(n=5,N=15,k=5,m=5,t=1;r=2)$ MCBC,
$$\begin{array}{|c|c|c|c|c|}
\hline
1 & 1 & 2 & 2 & 3 \\
3 & 4 & 3 & 4 & 4 \\
5 & 5 & 5 & 5 & 5 \\
\hline
\end{array}$$
where the $i$-th column contains the indices of items stored in the server $C_i\in\cC$, $i\in[5]$. It is possible to verify that the code $\cC$ satisfies the requirements of a $(5,15,5,5,1;2)$-MCBC. For example, the multiset request $\{3,3,4,4,5\}$ can be read by taking the subsets $D_1=\{3\}$, $D_2=\{4\}$, $D_3=\{3\}$, $D_4=\{4\}$, $D_5=\{5\}$.
\end{example}

Similarly to the original problem of combinatorial batch codes, the  goal in this paper is to minimize the total storage $N$ given the parameters $n,m,k,t$ and $r$ of an MCBC. Let $N(n,k,m,t;r)$ be the smallest $N$ such that an $(n,N,k,m,t;r)$-MCBC exists. An MCBC is called {\em optimal} if $N$ is minimal given $n,m,k,t,r$. In this paper, we focus on the case $t=1$, and thus omit $t$ from the notation and write it as an $(n,N,k,m;r)$-MCBC and its minimum storage by $N(n,k,m;r)$. In case $r=1$, i.e. an MCBC is a CBC, we further omit $r$ and write it as an $(n,N,k,m)$-CBC and its minimum storage as $N(n,k,m)$.



In \cite{PSW09}, the authors studied another class of CBCs, called \emph{uniform combinatorial batch codes} (\emph{uniform CBCs}), in which each item is stored in the same number of servers. Several constructions of optimal uniform CBCs were given in~\cite{BB14,BRR12,PSW09,SG13}. In this paper we consider a slightly different class of MCBCs, in which each server stores the same number of items, and call these codes {\em regular multiset combinatorial batch codes} ({\em regular MCBCs}).

A regular $(n,N,k,m;r)$-MCBC is an MCBC in which each server stores the same number $\mu$ of items, where $\mu = N/m$. Given $n,m,k,r$, let $\mu(n,k,m;r)$ denote the smallest number of items stored in each server, then the optimal value of $N$ is determined by $\mu(n,k,m;r)$, that is, $N=m\mu(n,k,m;r)$.

\subsection{Previous Results on CBCs}

For CBCs, a significant amount of work has been done to study the value $N(n,k,m)$, and the exact value has been determined for a large range of parameters. We list some of the known results below (for more details see~\cite{BRR12,BKMS10,BT11A,BT11B,PSW09,SG13}).

\begin{theorem}
\label{CBC}

\begin{itemize}
$~$
\item[(i)] $N(n,k,k)=kn-k(k-1)$.
\item[(ii)] If $n\geq (k-1){m\choose k-1}$, then $N(n,k,m)=kn-(k-1){m\choose k-1}$.
\item[(iii)] If ${m\choose k-2}\leq n\leq (k-1){m\choose k-1}$, then $N(n,k,m)=(k-1)n-\left\lfloor\frac{(k-1){m\choose k-1}-n}{m-k+1}\right\rfloor$.
\item[(iv)] If ${m\choose k-2}-(m-k+1)A(m,4,k-3)\leq n\leq {m\choose k-2}$, then $N(n,k,m)=(k-2)n-\left\lfloor\frac{2({m\choose k-2}-n)}{m-k+1}\right\rfloor$ for $0\leq ({m\choose k-2}-n)\mod (m-k+1)<\frac{m-k+1}{2}$.
\item[(v)] $N(m+1,k,m)=m+k$.
\item[(vi)] Let $k$ and $m$ be integers with $2\leq k\leq m$, then $$N(m+2,k,m)=\begin{cases}
m+k-2+\lceil 2\sqrt{k+1}\rceil & \text{ if } m+1-k\geq \lceil \sqrt{k+1}\rceil, \\
2m-2+\left\lceil 1+\frac{k+1}{m+1-k}\right\rceil & \text{ if } m+1-k< \lceil \sqrt{k+1}\rceil. \\
\end{cases}$$
\item[(vii)] For all integers $n\geq m\geq 3$, $N(n,3,m)=\begin{cases}
2n-m+\left\lfloor\frac{n-3}{m-2}\right\rfloor & \text{ if } n\leq m^2-m, \\
3n-m^2+m & \text{ if } n\geq m^2-m. \\
\end{cases}$
\item[(viii)] For all integers $n\geq m\geq 4$,
$$N(n,4,m)=\begin{cases}
n & \text{ if n=m,} \\
2n-m+\left\lceil\frac{1+\sqrt{8n-8m+1}}{2}\right\rceil & \text{ if $m<n\leq \frac{m^2+6m}{8}$ and $m$ is even} \\
 & \text{~~~ or if $m<n\leq \frac{m^2+4m+3}{8}$ and $m$ is odd,} \\
2n-m+\left\lceil\frac{5+\sqrt{8n-16m+25}}{2}\right\rceil & \text{ if $\frac{m^2+6m+8}{8}\leq n<{m\choose 2}$ and $m$ is even} \\
 & \text{~~~ or if $\frac{m^2+4m+11}{8}<n<{m\choose 2}$ and $m$ is odd,} \\
2n-\frac{m-1}{2} & \text{ if $n=\frac{m^2+4m+11}{8}$ and $m$ is odd,} \\
3n-\left\lfloor\frac{m^2}{2}-\frac{n-m}{m-3}\right\rfloor & \text{ if ${m\choose 2}\leq n<3{m\choose 3}$,} \\
4n-3{m\choose 3} & \text{ if $3{m\choose 3}\leq n$.}
\end{cases}$$
\item[(ix)] For any prime power $q\geq 3$, $N(q^2+q-1,q^2-q-1,q^2-q)=q^3-q$.
\end{itemize}
\end{theorem}

\subsection{Our Contributions}

From the definition of MCBCs, one can observe that $r\leq k\leq tm$ and $n\leq N$. If $m\geq nr$, the trivial construction where each server stores a single item is optimal. Therefore, we only consider the case $m<nr$.

In this paper, we study the properties of MCBCs, and give a necessary and sufficient condition for the existence of MCBCs. We give the following bounds on the value of $N(n,k,m;r)$.

\begin{theorem}\label{th:bounds}
\begin{enumerate}
$~$
\item[(i)] $N(n,k,m;r)\geq rn$.
\item[(ii)] $N(n,k,m;r)\geq N(n,k,m;i)$ for $i\in [r-1]$.
\item[(iii)] $\frac{1}{r}N(rn,k,m)\leq N(n,k,m;r)\leq N(rn,k,m)$.
\item[(iv)] $N(n,k,m;r)\leq rN(n,\left\lceil\frac{k}{r}\right\rceil,\left\lfloor\frac{m}{r}\right\rfloor)$.
\item[(v)] Let $r\leq k-1$. For any $c\in[r,k-1]$,
$N(n,k,m;r)\geq nc-\left\lfloor\frac{k-c}{m-k+1}\left[\frac{\left\lfloor\frac{k-1}{r}\right\rfloor{m\choose k-1}}{{m-c\choose k-1-c}}-n\right]\right\rfloor$.
\end{enumerate}
\end{theorem}

We also provide several constructions of $(n,N,k,m;r)$-MCBC and determine the exact value of $N(n,k,m;r)$ for some specific parameters.

\begin{theorem}\label{th:constructions}
\begin{enumerate}
$~$
\item[(i)] If $n\geq \left\lfloor\frac{k-1}{r}\right\rfloor{m\choose k-1}$, then $N(n,k,m;r)=kn-\left\lfloor\frac{k-1}{r}\right\rfloor{m\choose k-1}$.
\item[(ii)] $N(n,k,m;k)=kn$, $N(n,k,m;k-1)=\begin{cases}
kn-{m\choose k-1} & \text{if $n\geq {m\choose k-1}$}, \\
(k-1)n & \text{if $n< {m\choose k-1}$}.
\end{cases}$
\item[(iii)] If $\left\lfloor\frac{k-1}{r}\right\rfloor{m\choose k-1}-(m-k+1)A(m,4,k-2)\leq n\leq \left\lfloor\frac{k-1}{r}\right\rfloor{m\choose k-1}$ and $r\leq k-2$, then $N(n,k,m;r)=(k-1)n-\left\lfloor\frac{\left\lfloor\frac{k-1}{r}\right\rfloor{m\choose k-1}-n}{m-k+1}\right\rfloor$.
\item[(iv)] $N(n,k,k;r)=kn-\left\lfloor\frac{k-1}{r}\right\rfloor k$ if $r\mid k$, $n\geq \frac{k}{r}$ or $r\nmid k$, $n\geq \lfloor\frac{k}{r}\rfloor+r$.
\item[(v)] For any prime power $q$, $N(q^2+q,k,q^2;r)\leq q^3+q^2$, where $(k,r)$ satisfies $\lfloor\frac{q}{2}\rfloor+1\leq r\leq q$, $k\leq (q-r+1)(2r-1)$ or $r=1$, $k\leq q^2$. Especially, when $(k,r)\in\{(q^2,1),(2q-1,q)\}$, $N(q^2+q,k,q^2;r)=q^3+q^2$.
\end{enumerate}
\end{theorem}


For a regular $(n,N,k,m;k)$-MCBC, every item has to be stored in at least $k$ different servers and so $\mu(n,k,m;k)\geq kn/m$. Our contribution in this part is finding a necessary and sufficient condition for equality in the last inequality. This result is summarized in the following theorem.

\begin{theorem}\label{th:constructionsreg}
$\mu(n,k,m;k)=\frac{kn}{m}$ if and only if $n=c\cdot\frac{m}{\gcd{(m,k)}}$ for some integer $c\geq 0$.
\end{theorem}

The rest of the paper is organized as follows. In Section~\ref{sec:hallcondition}, we give a necessary and sufficient condition for the existence of MCBCs. In Sections~\ref{sec:bound} and \ref{sec:construction}, we give several bounds and constructions for MCBCs, and establish the results of $N(n,k,m;r)$ in Theorems~\ref{th:bounds} and \ref{th:constructions}. In Section~\ref{sec:regular_mcbc}, we analyse regular MCBCs, and determine the value of $\mu(n,k,m;k)$ in Theorem~\ref{th:constructionsreg}.

\section{Set Systems and the Multiset Hall's Condition}\label{sec:hallcondition}

A {\em set system} is a pair $(V,\cC)$, where $V$ is a finite set of {\em points} and $\cC$ is a collection of subsets of $V$ (called {\em blocks}). Given a set system $(V,\cC)$ with a points set $V=\{v_1,v_2,\dots,v_n\}$ and a blocks set $\cC=\{C_1,C_2,\dots,C_m\}$, its {\em incidence matrix} 
is an $m\times n$ matrix $M$, given by
$$M_{i,j}=\begin{cases}
1 & \text{if $v_j\in C_i$}, \\
0 & \text{if $v_j\not\in C_i$}.
\end{cases}$$
If $M$ is the incidence matrix of the set system $(V,\cC)$, then the set system having incidence matrix $M^\top$ is called the {\em dual set system} of $(V,\cC)$.

Let $\cC=\{C_1,C_2,\dots,C_m\}$ be an $(n,N,k,m;r)$-MCBC. Similarly to the study of CBCs, by setting $V=[n]$, we consider the set system $(V,\cC)$ of the MCBC. In addition, we denote the set system $(X,\cB)$ which is given by $X=[m]$ and $\cB=\{B_1,B_2,\dots,B_n\}$ where for each $i\in[n]$, $B_i\subseteq X$ consists of the servers that store the $i$-th item. Then, it is readily verified that $(X,\cB)$ is the dual set system of $(V,\cC)$. We note that a set system $(V,\cC)$ of this form or its dual set system $(X,\cB)$ uniquely determines an MCBC and thus in the rest of the paper we will usually refer to an MCBC by its set system or its dual set system.

\begin{example}
\label{e.ap4} The following is a $(20,80,16,16)$-CBC given in \cite{SG13} based on an {\em affine plane} of order $4$. Here, $V=[20]$, each column contains the indices of items stored in a server $C_i\in\cC$ and also forms a block of the set system $(V,\cC)$.
$$\begin{array}{|c|c|c|c|c|c|c|c|c|c|c|c|c|c|c|c|}
\hline
1 & 2 & 3 & 4 & 1 & 2 & 3 & 4 & 1 & 2 & 3 & 4 & 1 & 2 & 3 & 4 \\
5 & 6 & 7 & 8 & 6 & 5 & 8 & 7 & 7 & 8 & 5 & 6 & 8 & 7 & 6 & 5 \\
9 & 10 & 11 & 12 & 12 & 11 & 10 & 9 & 10 & 9 & 12 & 11 & 11 & 12 & 9 & 10 \\
13 & 14 & 15 & 16 & 15 & 16 & 13 & 14 & 16 & 15 & 14 & 13 & 14 & 13 & 16 & 15 \\
17 & 17 & 17 & 17 & 18 & 18 & 18 & 18 & 19 & 19 & 19 & 19 & 20 & 20 & 20 & 20 \\
\hline
\end{array}$$

The incidence matrix of the CBC given above is as follows, where the indices of nonzero entries in the $i$-th row, $i\in[16]$, correspond to the indices of items stored in the $i$-th server $C_i$.
\begin{equation*}
\left(
\begin{array}{ccccc}
1~0~0~0~1~0~0~0~1~0~0~0~1~0~0~0~1~0~0~0 \\
0~1~0~0~0~1~0~0~0~1~0~0~0~1~0~0~1~0~0~0 \\
0~0~1~0~0~0~1~0~0~0~1~0~0~0~1~0~1~0~0~0 \\
0~0~0~1~0~0~0~1~0~0~0~1~0~0~0~1~1~0~0~0 \\
1~0~0~0~0~1~0~0~0~0~0~1~0~0~1~0~0~1~0~0 \\
0~1~0~0~1~0~0~0~0~0~1~0~0~0~0~1~0~1~0~0 \\
0~0~1~0~0~0~0~1~0~1~0~0~1~0~0~0~0~1~0~0 \\
0~0~0~1~0~0~1~0~1~0~0~0~0~1~0~0~0~1~0~0 \\
1~0~0~0~0~0~1~0~0~1~0~0~0~0~0~1~0~0~1~0 \\
0~1~0~0~0~0~0~1~1~0~0~0~0~0~1~0~0~0~1~0 \\
0~0~1~0~1~0~0~0~0~0~0~1~0~1~0~0~0~0~1~0 \\
0~0~0~1~0~1~0~0~0~0~1~0~1~0~0~0~0~0~1~0 \\
1~0~0~0~0~0~0~1~0~0~1~0~0~1~0~0~0~0~0~1 \\
0~1~0~0~0~0~1~0~0~0~0~1~1~0~0~0~0~0~0~1 \\
0~0~1~0~0~1~0~0~1~0~0~0~0~0~0~1~0~0~0~1 \\
0~0~0~1~1~0~0~0~0~1~0~0~0~0~1~0~0~0~0~1 \\
\end{array}
\right)
\end{equation*}

Here, the dual set system is $X=[16]$ and $\cB=\{\{1,5,9,13\},\{2,6,10,14\},$ $\{3,7,11,15\},\{4,8,$ $12,16\},\{1,6,11,16\},\{2,5,12,15\},\{3,8,9,14\},\{4,7,10,13\},\{1,8,10,15\},\{2,7,9,16\},\{3,6,12,$ $13\},\{4,5,11,14\},\{1,7,12,14\},\{2,8,11,13\},\{3,5,10,16\},\{4,6,9,15\},\{1,2,3,4\},\{5,6,7,8\},$ $\{9,10,11,12\},\{13,14,15,16\}\}$.
\end{example}

In the rest of this section we let $(V,\cC)$ with $V=[n]$ and $\cC=\{C_1,C_2,\dots,C_m\}$ be a set system, and $(X,\cB)$ with $X=[m]$ and $\cB=\{B_1,B_2,\dots,B_n\}$ be its dual set system. The following theorem states a necessary and sufficient condition on the dual set system to form a construction of CBCs.
\begin{theorem}[\cite{PSW09}]
\label{hallc} The set system $(V,\cC)$ is an $(n,N,k,m)$-CBC if and only if its dual set system $(X,\cB)$ satisfies the following {\em Hall's condition}:
\\\vbox{for all $h \in [k]$, and any $h$ distinct blocks $B_{i_1},B_{i_2},\dots,B_{i_h}\in\cB$, $|\cup_{j=1}^h B_{i_j}|\geq h$.}
\end{theorem}

The Hall's condition was generalized in several ways, see e.g.~\cite{AK1990,M1967,MP1966}. For example, in \cite{BT12}, the authors explored the value of $N(n,k,m,t)$ for $t>1$ with a generalization of the Hall's condition, that is, the set system $(V,\cC)$ is an $(n,N,k,m,t)$-CBC if and only if its dual set system satisfies the {\em $(k,t)$-Hall's condition: for all $h\in [k]$, and any $h$ distinct blocks $B_{i_1},B_{i_2},\dots,B_{i_h}\in\cB$, $|\cup_{j=1}^h B_{i_j}|\geq h/t$.} In this paper, we present another generalization of the Hall's condition, named the {\em multiset Hall's condition}, and provide a necessary and sufficient condition for the construction of MCBCs.

\begin{theorem}
\label{mhallc}
The set system $(V,\cC)$ is an $(n,N,k,m;r)$-MCBC if and only if its dual set system $(X,\cB)$ satisfies the following {\em multiset Hall's condition}:
\\\vbox{for all $h\in\left[\lceil\frac{k}{r}\rceil\right]$, and any $h$ distinct blocks $B_{i_1},B_{i_2},\dots,B_{i_h}\in\cB$, $|\cup_{j=1}^h B_{i_j}|\geq \min\{hr,k\}$.}
\end{theorem}
\begin{proof}
($\Rightarrow$) Assume that $(V,\cC)$ is an $(n,N,k,m;r)$-MCBC, and let $i_1,i_2,\dots,i_h\in V$ for some $h\in\left[\lceil\frac{k}{r}\rceil\right]$ be the indices of some $h$ different items. Then, the set $\cup_{j\in[h]} B_{i_j}$ corresponds to the indices of all the servers that contain these items.

If $h\leq \lfloor\frac{k}{r}\rfloor$, let us consider the multiset request $\{i_1,\ldots,i_1,i_2,\ldots,i_2,\ldots, i_h,\ldots,i_h\}$ where each of the $h$ elements is requested $r$ times. Since it is possible to read from each server at most one item, the number of servers that contain these $h$ items has to be at least $hr$, that is $|\cup_{j\in[h]} B_{i_j}|\geq hr$. Similarly, if $h=\lceil\frac{k}{r}\rceil$, then we need $k$ servers for the multiset request of size $k$ on $i_1,i_2,\dots,i_h$ where each $i_j$, for $j\in[h]$, is requested at most $r$ times, and so $|\cup_{j\in[h]} B_{i_j}|\geq k$. Together we conclude that $|\cup_{j=1}^h B_{i_j}|\geq \min\{hr,k\}$.


($\Leftarrow$) We construct a new set system $(U,\cF)$ with $U=[rn]$ and $\cF=\{F_1,F_2,\dots,F_m\}$ where for $i\in [m]$, $F_i=\{c+jn:c\in C_i,j\in[0,r-1]\}$. Let $U^{(\ell)}=\{\ell+jn:j\in[0,r-1]\}$ for $\ell\in[n]$. We first show the following claim.
\begin{claim}\label{cl1}
$(V,\cC)$ is an $(n,N,k,m;r)$-MCBC if and only if $(U,\cF)$ is an $(rn,rN,k,m)$-CBC.
\end{claim}
\begin{proof}
Suppose that $(V,\cC)$ is an $(n,N,k,m;r)$-MCBC. For any request $P=\{i_1,i_2,\dots,i_k\}\subseteq[rn]$ of $(U,\cF)$, let $r_\ell$, for $\ell\in[n]$, be the number of elements requested from the set $U^{(\ell)}$, and note that $0\leq r_\ell\leq r$ and $\sum_{\ell=1}^nr_\ell=k$. Consider the multiset request $Q$ of $(V,\cC)$ where each $\ell\in [n]$ appears $r_\ell$ times. Since $(V,\cC)$ is an $(n,N,k,m;r)$-MCBC, $Q$ can be read by choosing subsets $D_j\subseteq C_j$, $|D_j|\leq 1$ for $j\in[m]$. Then $P$ can also be read from the set of servers $\{F_j:j\in[m],|D_j|=1\}$. Therefore, $(U,\cF)$ is an $(rn,rN,k,m)$-CBC.\footnote{We notice that this direction is not needed in the proof. But we still prove it here because we will use it in Lemma~\ref{bound} below.}

The reverse is similar. Suppose that $(U,\cF)$ is an $(rn,rN,k,m)$-CBC. For any multiset request $Q$ of $(V,\cC)$ where each $\ell\in [n]$ appears $r_\ell$ times, consider the request $P$ of $(U,\cF)$ which contain any $r_\ell$ distinct elements in $U^{(\ell)}$. Since $(U,\cF)$ is an $(rn,rN,k,m)$-CBC, $P$ can be read by taking $D_j\subseteq F_j$, $|D_j|\leq 1$ for $j\in[m]$. Then $Q$ can be read from the servers $\{C_j:j\in[m],|D_j|=1\}$. Therefore, $(V,\cC)$ is an $(n,N,k,m;r)$-MCBC.
\end{proof}

Let $(X =[m], \cG =\{G_1,\ldots,G_{nr}\})$, be the dual set system of  $(U,\cF)$, so $G_i$, for $i\in[rn]$, is the set of servers that contain the $i$-th item in $(U,\cF)$. We show that $(X,\cG)$ satisfies the Hall's condition. For any $i_1,i_2,\dots,i_h\in[rn]$, $h\in [k]$, let $r_\ell$ denote the number of elements in $U^{(\ell)}$ for $\ell\in[n]$. Then $\left|\cup_{j=1}^hG_{i_j}\right|=\left|\cup_{\ell:r_\ell\neq 0}B_\ell\right|$.

Let $a=|\{\ell:r_\ell\neq 0\}|$. By the multiset Hall's condition, when $a\leq \lceil\frac{k}{r}\rceil$, $\left|\cup_{\ell:r_\ell\neq 0}B_\ell\right|\geq \min\{ar,k\}$; when $\lceil\frac{k}{r}\rceil<a\leq k$, $$\left|\bigcup_{\ell:r_\ell\neq 0}B_\ell\right|\geq \min\{r\left\lceil\frac{k}{r}\right\rceil,k\}\geq k=\min\{ar,k\}.$$
Since $h=\sum_{\ell:r_\ell\neq 0}r_\ell\leq ar$ and $h\leq k$, we always have $|\cup_{j=1}^h G_{i_j}|\geq h$ for any $h\in[k]$, that is $(X,\cG)$ satisfies the Hall's condition. Hence, $(U,\cF)$ is an $(rn,rN,k,m)$-CBC by Theorem~\ref{hallc}, and by Claim~\ref{cl1} $(V,\cC)$ is an $(n,N,k,m;r)$-MCBC.
\end{proof}

Theorem~\ref{hallc} is a special case of Theorem~\ref{mhallc} for $r=1$. In the following, when constructing an MCBC, we always construct its dual set system $(X,\cB)$, and check if it satisfies the multiset Hall's condition from Theorem~\ref{mhallc}. By adding an asterisk, we let $(X,\cB)^\ast$ denote its dual set system $(V,\cC)$. The following properties are obvious.

\begin{remark}\label{basic}
\begin{enumerate}
$~$
\item[(i)] If there exists some $h_0<\lceil\frac{k}{r}\rceil$ such that for any $h_0$ blocks $B_{i_1},B_{i_2},\dots,B_{i_{h_0}}$, $|\cup_{j=1}^{h_0} B_{i_j}|\geq k$, then for any $h$ such that $h_0<h\leq \lceil\frac{k}{r}\rceil$, the multiset Hall's condition is also satisfied.
\item[(ii)] An $(n,N,k,m;r)$-MCBC is also an $(n,N,k',m;r)$-MCBC for any $k'\leq k$.
\end{enumerate}
\end{remark}

\begin{example}
By checking the multiset Hall's condition, it is possible to verify that Example~\ref{e.ap4} gives a construction of $(20,80,k,16;r)$-MCBC for any pair $(k,r)\in\{(16,1),(11,2),(10,3),(7,4)\}$. Especially, as will be shown in Construction~\ref{cons.ap} in the sequel, the code is optimal when $(k,r)\in\{(16,1),(7,4)\}$.
\end{example}

In the following sections, we will give several bounds and constructions of MCBCs.

\section{Bounds of MCBCs}\label{sec:bound}

In this section, we give several bounds of MCBCs, which provide the results stated in Theorem~\ref{th:bounds}.

\begin{lemma}
\label{trivial_bound}\begin{itemize}
$~$
\item[(i)] $N(n,k,m;r)\geq rn$.
\item[(ii)] $N(n,k,m;r)\geq N(n,k,m;i)$ for $i\in [r-1]$.
\item[(iii)] $N(n,k,m;k)=kn$.
\end{itemize}
\end{lemma}
\begin{proof}
(i) This inequality holds since each item has to be stored in at least $r$ servers. (ii) This inequality holds from the definition of MCBCs. (iii) By (i), $N(n,k,m;k)\geq kn$. The trivial construction where each item is stored in arbitrary $k$ servers gives an optimal code construction.
\end{proof}

\begin{lemma}
\label{recursive_bound}
\begin{itemize}
$~$
\item[(i)] $\frac{1}{r}N(nr,k,m)\leq N(n,k,m;r)\leq N(rn,k,m)$.
\item[(ii)] $N(n,k,m;r)\leq rN(n,\left\lceil\frac{k}{r}\right\rceil,\left\lfloor\frac{m}{r}\right\rfloor)$.
\end{itemize}
\end{lemma}

\begin{proof}
(i) From the proof of Claim~\ref{cl1} in Theorem~\ref{mhallc}, we can see that if there exists an $(n,N,k,m;r)$-MCBC, then there exists an $(rn,rN,k,m)$-CBC, and therefore $N(nr,k,m) \leq r N(n,k,m;r)$.

Assume that there exists an $(rn,N,k,m)$-CBC given by the set system $(U=[rn], \cF=\{F_1,F_2,\dots,F_m\})$. We construct a new set system $(V=[n],\cC=\{C_1,\ldots,C_m\})$ as follows.
The $i$-th server contains the items given by the set $C_i = \{(\ell-1) (\bmod  n)+1 :  \ell \in F_i\}$. 
That is, each item in the set $U^{(\ell)}=\{\ell+jn:j\in[0,r-1]\}$ in each server is replaced with $\ell$ for $\ell\in[n]$ (without repetitions). The new set system $(V,\cC)$ defines an $(n,N',k',m;r')$-MCBC with storage $N'\leq N$. To complete this proof we will show that $k'=k$ and $r'=r$.


Let $Q$ be a multiset request for $(V,\cC)$ where the $\ell$-th element, $\ell\in [n]$, is requested $r_\ell$ times, so $0\leq r_\ell\leq r$ and $\sum_{\ell=1}^nr_\ell=k$. Consider the request $P$ of $(U,\cF)$ which contains any $r_\ell$ distinct elements from $U^{(\ell)}$. Since $(U,\cF)$ is an $(rn,N,k,m)$-CBC, $P$ can be read by taking subsets $D_j\subseteq F_j$, $|D_j|\leq 1$ for $j\in[m]$. Then $Q$ can be read from the servers $\{C_j:j\in[m],|D_j|=1\}$. Hence, $(V,\cC)$ is an $(n,N',k,m;r)$-MCBC with $N'\leq N$, and $N(n,k,m;r)\leq N(rn,k,m)$.

(ii) Assume that there exists an $(n,N,\left\lceil\frac{k}{r}\right\rceil,\left\lfloor\frac{m}{r}\right\rfloor)$-CBC given by the set system $(V=[n], \cF=\{F_1,F_2,\ldots,F_{\lfloor\frac{m}{r}\rfloor}\})$. We construct a new code by the following set system $(V=[n], \cC=\{C_1,C_2,\dots,C_m\})$, such that $C_{i+j\lfloor\frac{m}{r}\rfloor}=F_i$ for any $i\in\left[\left\lfloor\frac{m}{r}\right\rfloor\right]$ and $j\in[0,r-1]$, and $C_\ell=\emptyset$ for any $r\lfloor\frac{m}{r}\rfloor+1 \leq \ell \leq m$. That is, each $F_i$ for $i\in\left[\left\lfloor\frac{m}{r}\right\rfloor\right]$ is repeated $r$ times.

Assume that the dual set system of $(V,\cF)$ is $(Y=\left[\left\lfloor\frac{m}{r}\right\rfloor\right],\cG=\{G_1,\dots,G_n\})$, and the dual set system of $(V,\cC)$ is $(X=[m],\cB=\{B_1,\dots,B_n\})$. Then $|B_i|=r|G_i|$ for $i\in[n]$. Since $(V,\cF)$ is an $(n,N,\left\lceil\frac{k}{r}\right\rceil,\left\lfloor\frac{m}{r}\right\rfloor)$-CBC, for any $1\leq h\leq \left\lceil\frac{k}{r}\right\rceil$, and distinct $i_1,\ldots,i_h\in [n]$, $|\cup_{j=1}^h G_{i_j}|\geq h$ by Theorem~\ref{hallc}. Then for any $1\leq h\leq \left\lceil\frac{k}{r}\right\rceil$, $|\cup_{j=1}^h B_{i_j}|\geq hr\geq \min\{hr,k\}$. Therefore, by Theorem~\ref{mhallc}, $(V,\cC)$ is an $(n,rN,k,m;r)$-MCBC, and $N(n,k,m;r)\leq rN(n,\left\lceil\frac{k}{r}\right\rceil,\left\lfloor\frac{m}{r}\right\rfloor)$.
\end{proof}

Let $(V,\cC)$ be a set system of an $(n,N,k,m;r)$-MCBC and let $(X,\cB)$ be its dual set system. For $i\geq 0$, we denote by $A_i$ the number of subsets in $\cB$ of size $i$. Note that for $i<r$, $A_i=0$ since every item is contained in at least $r$ different servers.
As pointed in \cite{PSW09}, $A_i=0$ for $i\geq k+1$ since for any block of size larger than $k$, we can reduce the block to $k$ points and the multiset Hall's condition is still satisfied. The following bound is a generalization of the results in \cite{BRR12,BT11A,PSW09}.

\begin{lemma}\label{lem:bound}
\label{b.sum} If $(X,\cB)^\ast$ is an $(n,N,k,m;r)$-MCBC with $r\leq k-1$, and $A_i$ for $i\in[k-1]$ is defined as above, then $$\sum_{i=r}^{k-1}{m-i\choose k-1-i}A_i\leq \left\lfloor\frac{k-1}{r}\right\rfloor{m\choose k-1}.$$
\end{lemma}
\begin{proof}
Let $M_{k-1}$ be the ${m\choose k-1}\times n$ matrix, whose rows are labeled by all the $(k-1)$-subsets of $X$, and the columns are labeled by the blocks in $\cB$ that contain less than $k$ points. The $(i,j)$-th entry of $M_{k-1}$ is $1$ if the $j$-th block $B_j$ is contained in the $i$-th $(k-1)$-subset of $X$, and otherwise it is $0$.

Each row in $M_{k-1}$ has at most $\left\lfloor\frac{k-1}{r}\right\rfloor$ ones. In order to verify this property, assume in the contrary that there exist $\left\lfloor\frac{k-1}{r}\right\rfloor+1$ blocks, and without loss of generality let them be the blocks $B_1,B_2,\dots,B_{\lfloor\frac{k-1}{r}\rfloor+1}$, which are all subsets of the same $(k-1)$-subset. Therefore, $|\cup_{i=1}^{\lfloor\frac{k-1}{r}\rfloor+1} B_i|\leq k-1$, and the multiset Hall's condition is not satisfied, since $\lfloor\frac{k-1}{r}\rfloor+1 = \lceil k/r \rceil$ and  $\min\{ (\lfloor\frac{k-1}{r}\rfloor+1)r,k \} = k$. Every column which corresponds to a block of size $i<k$ has exactly ${m-i\choose k-1-i}$ ones. Therefore, by counting the number of ones in $M_{k-1}$ by rows and columns separately, we get that $\sum_{i=r}^{k-1}{m-i\choose k-1-i}A_i\leq\left\lfloor\frac{k-1}{r}\right\rfloor{m\choose k-1}$.
\end{proof}

According to Lemma~\ref{lem:bound}, we derive the next theorem.
\begin{theorem}
\label{bound} Let $r\leq k-1$. For any $c\in[r,k-1]$,
$$N(n,k,m;r)\geq nc-\left\lfloor\frac{k-c}{m-k+1}\left[\frac{\left\lfloor\frac{k-1}{r}\right\rfloor{m\choose k-1}}{{m-c\choose k-1-c}}-n\right]\right\rfloor.$$
\end{theorem}
\begin{proof}
The proof is similar to the one given in Lemma 3.2 in~\cite{BRR12}, and hence we omit it here.
\end{proof}

\section{Constructions of MCBCs}\label{sec:construction}

In this section we present several constructions of MCBCs. Constructions~\ref{c.largen} and \ref{cons2} are generalizations of the equivalent ones in \cite{BRR12,BT11A,PSW09} which determine the value of $N(n,k,m)$ in Theorem~\ref{CBC} (ii) and (iii). Construction~\ref{cons.m=k} is a generalization of that in \cite{PSW09} which determines the value of $N(n,k,k)$ in Theorem~\ref{CBC} (i).

\subsection{A Construction by Replication}

Our first construction uses simple replication which is a generalization of the one in \cite{BRR12,BT11A,PSW09}.
\begin{construction}
\label{c.largen}
Let $n,k,m,r$ be positive integers such that $n\geq \left\lfloor\frac{k-1}{r}\right\rfloor{m\choose k-1}$ with $r<k$. We construct an $(n,N,k,m;r)$-MCBC with $N=kn-\left\lfloor\frac{k-1}{r}\right\rfloor{m\choose k-1}$, by explicitly constructing its dual set system $(X=[m],\cB=\{B_1,\ldots,B_n\})$ as follows:
\begin{enumerate}
\item The first $\left\lfloor\frac{k-1}{r}\right\rfloor{m\choose k-1}$ blocks of $\cB$ consist of $\left\lfloor\frac{k-1}{r}\right\rfloor$ copies of all different $(k-1)$-subsets of $[m]$.
\item Each remaining block of $\cB$ is taken to be any $k$-subset of $[m]$.
\end{enumerate}
Thus, the value of $N$ is given by
$$N=\left\lfloor\frac{k-1}{r}\right\rfloor{m\choose k-1} (k-1) + \left(n-  \left\lfloor\frac{k-1}{r}\right\rfloor{m\choose k-1}\right)k = kn -  \left\lfloor\frac{k-1}{r}\right\rfloor{m\choose k-1}.$$
\end{construction}
The correctness of this construction is proved in the next theorem.
\begin{theorem}
The code $(X,\cB)^*$ from Construction~\ref{c.largen} is an $(n,N,k,m;r)$-MCBC with $n\geq \left\lfloor\frac{k-1}{r}\right\rfloor{m\choose k-1}$, $r<k$ and $N=kn-\left\lfloor\frac{k-1}{r}\right\rfloor{m\choose k-1}$.
\end{theorem}
\begin{proof}
We only need to check that $(X,\cB)$ satisfies the multiset Hall's condition. For $1\leq h\leq \lceil\frac{k}{r}\rceil$, let $B_{i_1},B_{i_2},\dots,B_{i_h}\in\cB$ be some $h$ different blocks. If there exists a block of size $k$ or there exist two distinct blocks of size $k-1$, then $|\sum_{j=1}^n B_{i_j}|\geq k$; otherwise, we have $h\leq \left\lfloor\frac{k-1}{r}\right\rfloor$ by construction, and $|\sum_{j=1}^n B_{i_j}|\geq k-1\geq \min\{hr,k\}$.
\end{proof}
Before we show that this construction is optimal, let us recall a useful lemma from~\cite{BRR12,BT11A}.
\begin{lemma}[\cite{BRR12,BT11A}]
\label{ineq} Let $1\leq k\leq m$ and $0\leq i\leq k-1$. Then ${m-i\choose k-1-i}-1\geq (m-k+1)(k-1-i)$.
\end{lemma}
We can now deduce that Construction~\ref{c.largen} is optimal.
\begin{cor}
\label{b.largen} For any $n\geq \left\lfloor\frac{k-1}{r}\right\rfloor{m\choose k-1}$, $N(n,k,m;r)=kn-\left\lfloor\frac{k-1}{r}\right\rfloor{m\choose k-1}$.
\end{cor}
\begin{proof}
For any $(n,N,k,m;r)$-MCBC, let $A_i$ for $i\in[k]$ be the number of blocks in the dual set system of size $i$. By Lemma~\ref{ineq}, for $i\leq k-1$, ${m-i\choose k-1-i}\geq (m-k+1)(k-1-i)+1\geq k-i$, then $$\sum_{i=r}^{k-1}(k-i)A_i\leq \sum_{i=r}^{k-1}{m-i\choose k-1-i}A_i\leq \left\lfloor\frac{k-1}{r}\right\rfloor{m\choose k-1},$$
and the last inequality holds according to Lemma~\ref{b.sum}. Therefore, we get that
\begin{align*}
N & =\sum_{i=r}^k iA_i = \sum_{i=r}^k (k-(k-i))A_i= \sum_{i=r}^k kA_i - \sum_{i=r}^k (k-i)A_i \\
 & = kn - \sum_{i=r}^{k-1} (k-i)A_i\geq kn-\left\lfloor\frac{k-1}{r}\right\rfloor{m\choose k-1}.
\end{align*}
Hence, we conclude that $N(n,k,m;r)=kn-\left\lfloor\frac{k-1}{r}\right\rfloor{m\choose k-1}$ when $n\geq \left\lfloor\frac{k-1}{r}\right\rfloor{m\choose k-1}$, since the codes from Construction~\ref{c.largen} achieve this bound.
\end{proof}
As a special case when $r=k-1$ we get the following corollary.
\begin{cor}
\label{r=k-1} $N(n,k,m;k-1)=\begin{cases}
kn-{m\choose k-1} & \text{if $n\geq {m\choose k-1}$}, \\
(k-1)n & \text{if $n< {m\choose k-1}$}. \\
\end{cases}$
\end{cor}
\begin{proof}
For $n\geq {m\choose k-1}$, according to Corollary~\ref{b.largen} for $r=k-1$, we get that $N(n,k,m;k-1) = kn-{m\choose k-1}$. For $n<{m\choose k-1}$, we slightly modify the code from Construction~\ref{c.largen} such that the $n$ blocks in $\cB$ are some different $(k-1)$-subsets of $[m]$. It is readily verifies that the multiset Hall's condition holds for this modified construction and thus it provides an $(n,N=n(k-1),k,m;k-1)$-MCBC. Finally, according to Lemma~\ref{trivial_bound} (i), this construction is optimal.
\end{proof}

\subsection{Constructions Based on Constant Weight Codes}
Next, we give constructions based upon constant weight codes. Let $(n,d,w)$-code denote a binary constant weight code of length $n$, weight $w$ and minimum Hamming distance $d$, and let $A(n,d,w)$ denote the maximum number of codewords of an $(n,d,w)$-code.

\begin{construction}
\label{cons.smalln} Let $X=[m]$ and $\cC$ be an $(m,2(k-w),w)$-code with $n$ codewords for some $w\in[r,k-1]$. Let $\cB=\{B_1,\ldots,B_n\}$ be the support sets of all the codewords in $\cC$.
\end{construction}
\begin{theorem}
The code $(X,\cB)^*$ from Construction~\ref{cons.smalln} is an $(n,wn,k,m;r)$-MCBC.
\end{theorem}
The next theorem proves the correctness of this construction.
\begin{proof}
We only need to check that $(X,\cB)$ satisfies the multiset Hall's condition. It is satisfied as the size of each block in $\cB$ is $w\geq r$ and since the minimum distance of $\cC$ is $2(k-w)$, we get that the union of any two blocks in $\cB$ is at least $k$.
\end{proof}

If we take $w=r$, we get the following family of optimal codes.
\begin{cor}
\label{nsmall} For any $n\leq A(m,2(k-r),r)$, $N(n,k,m;r)=rn$.
\end{cor}
\begin{proof}
By Construction~\ref{cons.smalln}, if there exists an $(m,2(k-r),r)$-code with $A(m,2(k-r),r)$ codewords, we get an $(n,rn,k,m;r)$-MCBC for any $n\leq A(m,2(k-r),r)$, and it is optimal by Lemma~\ref{trivial_bound} (i).
\end{proof}

Constant weight codes are used in \cite{BRR12} to prove Theorem~\ref{CBC} (iv). Now, we give a similar construction for MCBCs.

\begin{construction}
\label{cons2} Let $X=[m]$, $r\leq k-2$. Let $\cC$ be an $(m,4,k-2)$-code with $\alpha$ codewords with $\alpha\leq A(k,4,k-2)$. First, let $\cB_0$ be a set of $\left\lfloor\frac{k-1}{r}\right\rfloor{m\choose k-1}$ blocks, in which each $(k-1)$-subset of $[m]$ appears $\left\lfloor\frac{k-1}{r}\right\rfloor$ times. Let $\cS$ consist of the support sets of the codewords in $\cC$. Then, for any block in $\cS$, add it to $\cB_0$, and remove one copy of each of its $m-k+2$ supersets\footnote{For a block $S\in\cS$ of size $k-2$, the {\em supersets} are the $(k-1)$-subsets of $[m]$ that contain $S$.} of size $k-1$ in $\cB_0$. Let the resulting block set be $\cB$.
\end{construction}
\begin{theorem}
The code $(X,\cB)^*$ from Construction~\ref{cons2} is an $(n,N,k,m;r)$-MCBC with $$n=\left\lfloor\frac{k-1}{r}\right\rfloor{m\choose k-1}-\alpha(m-k+1)\text{ and }N=n(k-1)-\alpha,$$
where $\alpha\leq A(k,4,k-2)$.
\end{theorem}
\begin{proof}
Since the code has minimum distance four, for any two blocks in $\cS$, their supersets of size $k-1$ are different. Therefore, each $(k-1)$-subset of $[m]$ is removed at most once. During the process, we add $\alpha$ blocks and remove $\alpha(m-k+2)$ blocks. Hence, we get that $n=\left\lfloor\frac{k-1}{r}\right\rfloor{m\choose k-1}-\alpha(m-k+1)$. Finally, since only $\alpha$ of all $n$ blocks are of size $k-2$, we get that $N=n(k-1)-\alpha$.  

Next we show that the multiset Hall's condition holds. For $1\leq h\leq \lceil\frac{k}{r}\rceil$, let $B_{i_1},B_{i_2},\dots,B_{i_h}\in\cB$ be some $h$ different blocks. In case $h=1$, the size of each block in $\cB$ is at least $k-2\geq r$ so the condition holds and thus we assume that $h\geq 2$.
If there exist two blocks $B_{i_a},B_{i_b}\in\cB\setminus\cS$ which are different $(k-1)$-subsets, then $|B_{i_a}\cup B_{i_b}|\geq k$; if there exists two blocks $B_{i_a},B_{i_b}\in \cS$, then $|B_{i_a}\cup B_{i_b}|\geq k$ because of the minimum Hamming distance of the code $\cC$ is four.

Therefore, we only need to check the case when there is one block from $\cS$, and the other $h-1$ blocks are the same $(k-1)$-subset from $\cB\setminus\cS$. For example, $B_{i_1}\in\cS$, and $B_{i_2},\dots,B_{i_h}\in\cB\setminus\cS$ are the same $(k-1)$-subset. If $B_{i_1}$ is not a subset of $B_{i_2}$, then $|B_{i_1}\cup B_{i_2}|\geq k$; if $B_{i_1}$ is a subset of $B_{i_2}$, then by the construction $h\leq\left\lfloor\frac{k-1}{r}\right\rfloor$ and $|\cup_{j=1}^h B_{i_j}|=k-1\geq hr=\min\{hr,k\}$. Therefore, the multiset Hall's condition holds.
\end{proof}

The following lower bound of $A(n,4,w)$ is known.
\begin{lemma}[\cite{GS1980}]
\label{cwcd=4} $A(n,4,w)\geq\frac{1}{n}{n\choose w}$.
\end{lemma}
Next, we apply Construction~\ref{cons2} to get a family optimal codes.
\begin{cor}\label{cor:r=k-2}
For any $\left\lfloor\frac{k-1}{r}\right\rfloor{m\choose k-1}-(m-k+1)A(m,4,k-2)\leq n\leq \left\lfloor\frac{k-1}{r}\right\rfloor{m\choose k-1}$, $r\leq k-2$, $$N(n,k,m;r)=n(k-1)-\left\lfloor\frac{\left\lfloor\frac{k-1}{r}\right\rfloor{m\choose k-1}-n}{m-k+1}\right\rfloor.$$
\end{cor}

\begin{proof}
When $r\leq k-2$, taking $c=k-1$ in Theorem~\ref{bound}, we have
$N\geq n(k-1)-\left\lfloor\frac{\left\lfloor\frac{k-1}{r}\right\rfloor{m\choose k-1}-n}{m-k+1}\right\rfloor$.

For any positive integers $n,m,k,r$ such that $r+2\leq k\leq m$ and $$\left\lfloor\frac{k-1}{r}\right\rfloor{m\choose k-1}-(m-k+1)A(m,4,k-2)\leq n\leq \left\lfloor\frac{k-1}{r}\right\rfloor{m\choose k-1},$$ we have that $0\leq \left\lfloor\frac{\left\lfloor\frac{k-1}{r}\right\rfloor{m\choose k-1}-n}{m-k+1}\right\rfloor\leq A(m,4,k-2)$. Let $\alpha=\left\lfloor\frac{\left\lfloor\frac{k-1}{r}\right\rfloor{m\choose k-1}-n}{m-k+1}\right\rfloor$. By Construction~\ref{cons2}, there exists an $(n',N',k,m;r)$-MCBC with $$n'=\left\lfloor\frac{k-1}{r}\right\rfloor{m\choose k-1}-\alpha(m-k+1)\text{ and }N'=n'(k-1)-\alpha.$$ Removing any $n'-n$ blocks of size $k-1$ from its dual set system, we get an optimal $(n,N,k,m;r)$-MCBC with $N=N'-(k-1)(n'-n)=n(k-1)-\alpha$.
\end{proof}

\subsection{A Construction for $m=k$}

In the following, we give a construction of $(n,N,k,k;r)$-MCBC and determine the value of $N(n,k,k;r)$ for $1\leq r\leq k$.
\begin{construction}
\label{cons.m=k}
Let $m=k$ and $X=[k],\cB=\{B_1,\ldots,B_n\}$ and $k=\alpha r+\beta$, where $\alpha\geq 1$ and $0\leq \beta\leq r-1$ such that the following holds.
\begin{enumerate}
\item[(i)] When $\beta=0$, for any $n\geq \alpha$, let $B_i=[(i-1)r+1,(i-1)r+r]$ for $i\in [\alpha]$, and $B_i=[k]$ for any $i\in[\alpha+1,n]$.
\item[(ii)] When $\beta>0$, for any $n\geq \alpha+r$, let $B_i=[(i-1)r+1,(i-1)r+r]$ for $i\in [\alpha]$, $B_i=[k]\setminus\{i-\alpha,i-\alpha+r,i-\alpha+2r,\dots,i-\alpha+(\alpha-1)r\}$ for $i\in[\alpha+1,\alpha+r]$, and $B_i=[k]$ for any $i\in[\alpha+r+1,n]$.
\end{enumerate}
\end{construction}
\begin{theorem}
The code $(X,\cB)^*$ from Construction~\ref{cons.m=k} is an $(n,N,k,k;r)$-MCBC with $N=kn-\left\lfloor\frac{k-1}{r}\right\rfloor k$.
\end{theorem}
\begin{proof}
We show that in both cases, i.e. $\beta>0$ and $\beta=0$, the multiset Hall's condition holds.\\
(i) $\beta=0$. For any $h\geq 1$ different blocks $B_{i_1},B_{i_2},\dots,B_{i_h}\in\cB$, if there exists some block $B_{i_j}=[k]$, then $|\cup_{j=1}^h B_{i_j}|=k\geq r$; otherwise, $i_j\in[\alpha]$ for all $j\in[h]$, and then $|\cup_{j=1}^h B_{i_j}|=hr$ since by the construction the blocks $B_1,\ldots,B_\alpha$ are mutually disjoint. Thus, the multiset Hall's condition holds, and it is an $(n,N,k,k;r)$-MCBC with
$$N=\alpha r + (n-\alpha)k= kn-\frac{k}{r}(k-r)=kn-\left\lfloor\frac{k-1}{r}\right\rfloor k.$$

(ii) $\beta>0$. First, note that for any $i\in[\alpha+1,\alpha+r]$, $|B_i|=k-\alpha\geq r$. This holds since $k-\alpha-r=\alpha r+\beta-\alpha-r=(\alpha-1)(r-1)+(\beta-1)\geq 0$. For any $h$ different blocks $B_{i_1},B_{i_2},\dots,B_{i_h}\in \cB$ with $h\geq 2$, if there exists some $j\in[h]$ such that $B_{i_j}=[k]$ or for all $j\in[h]$, $i_j\in[\alpha]$, then the proof is similar as in case (i). If there exist two blocks $B_{i_a},B_{i_b}$ such that $i_a,i_b\in[\alpha+1,\alpha+r]$, then $|B_{i_a}\cup B_{i_b}|=k$. Therefore, the remaining case to check is when only one block is from the set $\{B_i:i\in[\alpha+1,\alpha+r]\}$, and the other blocks are from the set $\{B_i:i\in[\alpha]\}$.
Without loss of generality assume that $i_1,i_2,\dots,i_{h-1}\in[\alpha]$ and $i_h\in[\alpha+1,\alpha+r]$. Since $|B_{i_h}|=k-\alpha$, when $1\leq h\leq \alpha-1$, by the construction we get that $|\cup_{j=1}^h B_{i_j}|\geq k-\alpha+h-1$. Since $$k-\alpha+h-1-hr=\alpha r+\beta-\alpha+h-1-hr=(\alpha-h)(r-1)+(\beta-1)\geq 0,$$ we conclude that $|\cup_{j=1}^h B_{i_j}|\geq hr$. If $h=\alpha$, by the construction we get that $|\cup_{j=1}^h B_{i_j}|=k-1=\alpha r+\beta-1\geq \alpha r$. Lastly, if $h=\alpha+1$, then $|\cup_{j=1}^h B_{i_j}|=k$. Therefore, the multiset Hall's condition is satisfied for any $1\leq h\leq \left\lceil\frac{k}{r}\right\rceil$, and the code is an $(n,N,k,k;r)$-MCBC with
$$N=kn-\alpha (k-r)-\alpha r=kn-\alpha k=kn-\left\lfloor\frac{k-1}{r}\right\rfloor k.$$
\end{proof}

The next corollary summarizes the construction and results in this section.
\begin{cor}
\label{k=m} $N(n,k,k;r)=kn-\left\lfloor\frac{k-1}{r}\right\rfloor k$ if $r\mid k$, $n\geq \frac{k}{r}$ or $r\nmid k$, $n\geq \lfloor\frac{k}{r}\rfloor+r$.
\end{cor}
\begin{proof}
Taking $m=k$ in Lemma~\ref{b.sum}, we have that $\sum_{i=r}^{k-1}(k-i)A_i\leq \left\lfloor\frac{k-1}{r}\right\rfloor k$. Similarly to the proof of Corollary~\ref{b.largen}, we get that
\begin{align*}
N(n,k,k;r) & =\sum_{i=r}^k iA_i = \sum_{i=r}^k (k-(k-i))A_i= \sum_{i=r}^k kA_i - \sum_{i=r}^k (k-i)A_i \\
 & =kn-\sum_{i=r}^{k-1}(k-i)A_i\geq kn-\left\lfloor\frac{k-1}{r}\right\rfloor k.
\end{align*}
Hence we conclude that $N(n,k,k;r)=kn-\left\lfloor\frac{k-1}{r}\right\rfloor k$ since Construction~\ref{cons.m=k} gives optimal codes that reach this bound.
\end{proof}

\subsection{A Construction from Steiner Systems}

In the following we construct a class of MCBCs based upon Steiner systems, which is a generalization of Example~\ref{e.ap4}.

A \emph{Steiner system} $S(2,\ell,m)$ is a set system $(X,\cB)$, where $X$ is a set of $m$ points, $\cB$ is a collection of $\ell$-subsets (blocks) of $X$, such that each pair of points in $X$ occurs together in exactly one block of~$\cB$. By the well known Fisher's inequality \cite{MR2007}, for an $S(2,\ell,m)$ with $m>\ell\geq 2$, $|\cB|\geq m$. For the existence of Steiner systems, we refer the reader to~\cite{CM2007}.

\begin{theorem}
\label{c.steiner} Let $(X,\cB)$ be an $S(2,\ell,m)$ with $m> \ell$. Then $(X,\cB)^*$ is a $(|\cB|,\ell|\cB|,k,m;r)$-MCBC for any $\lfloor\frac{\ell}{2}\rfloor+1\leq r\leq \ell$ and $k\leq (\ell-r+1)(2r-1)$.
\end{theorem}
\begin{proof}
By Remark~\ref{basic} (ii), we only need to check that $(X,\cB)^*$ is a $(|\cB|,\ell|\cB|,k,m;r)$-MCBC for any $\lfloor\frac{\ell}{2}\rfloor+1\leq r\leq \ell$ and $k=(\ell-r+1)(2r-1)$.

Let us first determine the number of points in the union of any $h$ blocks in $\cB$. Since any two blocks intersect in at most one point, for any $i\in[2,h]$, if the first $i-1$ blocks are chosen, then the $i$-th block can contribute at least $\ell-(i-1)$ new points. Therefore, if $h\leq \ell+1\leq |\cB|$, the union of any $h$ blocks contains at least
$$\ell+(\ell-1)+\cdots+(\ell-(h-1))=h\ell-{h\choose 2}$$ points.

Let us consider some $h$ blocks $B_{i_1},B_{i_2},\dots,B_{i_h}$ with
$$1\leq h\leq \left\lceil\frac{k}{r}\right\rceil=\left\lceil\frac{(\ell-r+1)(2r-1)}{r}\right\rceil\leq \left\lceil 2(\ell-r+1)-\frac{l+1}{r}+1\right\rceil\leq 2(\ell-r)+2,$$
where the last inequality holds since $r\geq\lfloor\frac{\ell}{2}\rfloor+1$. We see $2(\ell-r)+2\leq \ell+1$ for $r\geq\lfloor\frac{\ell}{2}\rfloor+1$.

If $h\in[2(\ell-r)+1]$, then $r\leq \ell-\frac{h-1}{2}$, and
$$|\cup_{j=1}^h B_{i_j}|\geq h\ell-{h\choose 2}=h\left(\ell-\frac{h-1}{2}\right)\geq hr=\min\{hr,k\}.$$
If $h=2(\ell-r)+2$, then $|\cup_{j=1}^h B_{i_j}|\geq \ell h-{h\choose 2}=(\ell-r+1)(2r-1)=k$. Therefore, the multiset Hall's condition holds for any $1\leq h\leq \lceil\frac{k}{r}\rceil$.
\end{proof}

An \emph{affine plane} of order $q$ is an $S(2,q,q^2)$. It has $q^2$ points and $q^2+q$ blocks. It is well known that an affine plane exists for any prime power $q$ \cite{CM2007}. The next result of CBCs based upon affine planes was given in~\cite{SG13}.


\begin{theorem}[\cite{SG13}]\label{cons.ap}
Let $q$ be a prime power and $(X,\cB)$ be an affine plane of order $q$. Then $(X,\cB)^*$ is an optimal uniform $(q^2+q,q^3+q^2,q^2,q^2)$-CBC.
\end{theorem}

The code in Theorem~\ref{cons.ap} is also an optimal CBC, since for $k=m$ by Theorem~\ref{CBC} (i), $N(q^2+q,q^2,q^2)=q^3+q^2$. However, note that it is a different code from the optimal $(n,N,k,k)$-CBC in \cite{PSW09} which is constructed as follows: Let $X=[k]$, and $\cB=\{B_1,\ldots,B_n\}$, which are given by $B_i=\{i\}$ for $i\in[k]$, and $B_i=[k]$ for $i\in[k+1,n]$. By Theorem~\ref{c.steiner}, we can see the code in Theorem~\ref{cons.ap} is also $(q^2+q,q^3+q^2,k,q^2;r)$-MCBCs for different pair-values of $k$ and $r$.

\begin{cor}
\label{c.ap} Let $q$ be a prime power. Then there exists a $(q^2+q,q^3+q^2,k,q^2;r)$-MCBC for any $\lfloor\frac{q}{2}\rfloor+1\leq r\leq q$ and $k\leq (q-r+1)(2r-1)$.
\end{cor}

When $r=q$, we also receive an optimal $(q^2+q,q^3+q^2,2q-1,q^2;q)$-MCBC, since it reaches the bound in Lemma~\ref{trivial_bound} (i) with total storage $N=rn=q(q^2+q)=q^3+q^2$. Note that this code could also be obtained by Construction~\ref{cons.smalln} using $(q^2,2q-2,q)$-codes. The existence of $(q^2,2q-2,q)$-codes follows from affine planes as follows. For any block $B\in\cB$, we get a codeword ${\sf u}$ of length $q^2$ in which the value of each coordinate ${\sf u}_i$ for $i\in[q^2]$ is $1$ if and only if $i\in B$. Since any two blocks intersect in at most one points, the distance between every two distinct codewords is at least $2(q-1)$. Lastly, we note that it is possible to improve the value of $k$ when $r\leq\lfloor\frac{q}{2}\rfloor$.
We demonstrate this in the following example.
\begin{example}
Let $q=4$. From Corollary~\ref{c.ap}, we obtain a $(20,80,k,16;r)$-MCBC for $(k,r)\in\{(10,3),(7,4)\}$. From the incidence matrix in Example~\ref{e.ap4}, the lower bounds on the size of the union of any $h$ blocks, $1\leq h\leq 6$, are shown in the following table:
\begin{center}
\begin{tabular}{|c|c|}
  \hline
  $h$ & size of union \\
  \hline
  1 & 4 \\
  2 & 7 \\
  3 & 9 \\
  4 & 10 \\
  5 & 10 \\
  6 & 11 \\
  \hline
\end{tabular}
\end{center}
Therefore, we get also a $(20,80,11,16;2)$-MCBC.
\end{example}

\section{Regular MCBCs}\label{sec:regular_mcbc}

In this section, we study regular MCBCs, and give a construction for such codes. Given $n,m,k,r$, let $\mu(n,k,m;r)$ denote the smallest number of items stored in each server in a regular MCBC. The following lemma presents a simple lower bound on the value of $\mu(n,k,m;r)$.
\begin{lemma}
\label{bound_reg} $\mu(n,k,m;r)\geq\left\lceil\frac{N(n,k,m;r)}{m}\right\rceil$.
\end{lemma}
\begin{proof}
This property holds since a regular $(n,N,k,m;r)$-MCBC is also an $(n,N,k,m;r)$-MCBC, and therefore $m\mu(n,k,m;r)\geq N(n,k,m;r)$.
\end{proof}

\begin{remark}
It is easy to check that the constructions of MCBCs in Section~\ref{sec:construction} also give regular MCBCs for some specific parameters. Moreover, when the MCBCs are optimal, the bound in Lemma~\ref{bound_reg} holds with equality.
\end{remark}

For $r=k$, we determine when regular MCBCs with minimum storage $kn$ exist.
\begin{construction}
\label{reg_cons} Let $n=\frac{m}{\gcd{(m,k)}}$ and $k\leq m$, then we have $\frac{nk}{m}=\frac{k}{\gcd{(m,k)}}$ and $m|nk$. Let $I=[0,nk-1]\subseteq \mathbb{Z}$, and for each $i\in[n]$, $I^{(i)}=[(i-1)k,ik-1]$. Then let $X=[m]$ and $\cB=\{B_1,\dots,B_n\}$, where $B_i=\{j\pmod{m}+1:j\in I^{(i)}\}$ for $i\in[n]$.
\end{construction}
\begin{theorem}
The code $(X,\cB)^*$ from Construction~\ref{reg_cons} is a regular $(n,k,m;k)$-MCBC.
\end{theorem}
\begin{proof}
For any $a,b\in I$, $a\neq b$, if $a\pmod{m}+1=b\pmod{m}+1$, then $m\mid(a-b)$. Since each $I^{(i)}$ consists of $k$ consecutive integers and $k\leq m$,
then $|B_i|=k$ for any $i\in[n]$. We only need to prove that the code is regular. This holds since each $i\in[m]$ appears in exactly $\frac{nk}{m}=\frac{k}{\gcd{(m,k)}}$ blocks in $\cB$.
\end{proof}

\begin{cor}
\label{rmcbc} $\mu(n,k,m;k)=\frac{kn}{m}$ if and only if $n=c\cdot\frac{m}{\gcd{(m,k)}}$ for some integer $c\geq 0$.
\end{cor}
\begin{proof}
If $\mu(n,k,m;k)=\frac{kn}{m}$, then the value of $n$ satisfies $n=c\cdot\frac{m}{\gcd{(m,k)}}$ for some integer $c\geq 0$ so that $\mu(n,k,m;k)$ is an integer. For any such $n$, Construction~\ref{reg_cons} gives a code with the desired parameters as follows. Assume that $n=c\cdot\frac{m}{\gcd{(m,k)}}$, and let $X=[m]$ and $\hat{\cB}$ consist of $c$ copies of all the blocks of $\cB$ in Construction~\ref{reg_cons}. Then, $(X,\hat{\cB})^*$ is a regular $(n,k,m;k)$-MCBC.
\end{proof}

\section{Conclusion}\label{sec:conclusion}

In this paper, we generalized combinatorial batch codes to multiset combinatorial batch codes and regular multiset combinatorial batch codes. Several bounds and constructions of optimal codes were obtained.

\section*{Acknowledgments}

The authors would like to thank Prof. Tuvi Etzion for valuable discussions.

\end{document}